\documentclass[a4paper,onecolumn,accepted=2025-11-06]{quantumarticle}
\pdfoutput=1 
\usepackage[utf8]{inputenc}
\usepackage[english]{babel}
\usepackage[T1]{fontenc}

\usepackage{amsmath}
\usepackage{amsfonts}
\usepackage{amsthm}  
\usepackage{hyperref} 
\usepackage{cleveref} 
\usepackage{xcolor}
\usepackage{tikz}

\usepackage[mathscr]{eucal}
\usepackage[caption=false]{subfig}
\usepackage[export]{adjustbox}

\newtheorem{myprop}{Proposition}
\newtheorem{mytheorem}{Theorem}
\crefname{mytheorem}{theorem}{theorems}
\newtheorem{mylemma}{Lemma}

\newtheorem{mycor}{Corollary}
\newtheorem{myclaim}{Claim}

\theoremstyle{definition}
\newtheorem{mydef}{Definition}
\Crefname{mydef}{Definition}{Definitions} 
\theoremstyle{plain}

\newcommand{\bra}[1]{\left\langle#1\right|}
\newcommand{\ket}[1]{\left|#1\right\rangle}
\newcommand{\tr}{\operatorname{tr}}
\newcommand{\h}{\mathcal{H}}
\newcommand{\D}{\mathcal{D}}

\newcommand{\Sp}{\operatorname{S}}
\newcommand{\pperp}{\scriptscriptstyle{\perp}}
\newcommand{\Suv}[1]{\ket{s_{#1}\scriptstyle(u,v)}}
\newcommand{\Suvbra}[1]{\bra{s_{#1}\scriptstyle(u,v)}}
\newcommand{\Srs}[1]{\ket{s_{#1}\scriptstyle(\rho,\sigma)}}
\newcommand{\Srsbra}[1]{\bra{s_{#1}\scriptstyle(\rho,\sigma)}}
\newcommand{\Sxy}[1]{\ket{s_{#1}\scriptstyle(x,y)}}
\newcommand{\Sxybra}[1]{\bra{s_{#1}\scriptstyle(x,y)}}
\newcommand{\vv}{\operatorname{v}}

\title{Topologically driven no-superposing theorem with a tight error bound}
\author{Zuzana Gavorov\'a}
\affil{Rachel and Selim Benin
School of Computer Science and Engineering, The Hebrew University of Jerusalem,
91904 Jerusalem, Israel}
\affil{F\'isica Te\`orica: Informaci\'o i Fen\`omens Qu\`antics, Departament de F\'isica, Universitat Aut\`onoma de Barcelona, 08193 Bellaterra (Barcelona), Spain}
\date{} 

\begin{document}
\maketitle

\begin{abstract}
To better understand quantum computation we can search for its limits or no-gos, especially if analogous limits do not appear in classical computation. Classical computation easily implements and extensively employs the addition of two bit strings, so here we study ‘quantum addition’: the superposition of two quantum states. We prove the impossibility of superposing two unknown states, no matter how many samples of each state are available. The proof uses topology; a quantum algorithm of any sample complexity corresponds to a continuous function, but the function required by the superposition task cannot be continuous by topological arguments. Our result for the first time quantifies the approximation error and the sample complexity $N$ of the superposition task, and it is tight. We present a trivial algorithm with a large approximation error and $N=1$, and the matching impossibility of any smaller approximation error for any $N$. Consequently, our results limit state tomography as a useful subroutine for the superposition. State tomography is useful only in a model that tolerates randomness in the superposed state. The optimal protocol in this random model remains open.
\end{abstract}

\section{Introduction}

Quantum computing promises remarkable technological advances, some based on discoveries that were initially seen as limitations of quantum mechanics. The no-cloning theorem \cite{no_cloning} is an example of an uncovered limitation that turned into an advantage. It establishes the security of quantum cryptography, starting with quantum key distribution \cite{bennett1984proceedings,Scarani_2005}, and deepens our understanding of quantum mechanics all the way to its first principles; the first principles are often derived from, or checked against, the no-cloning theorem \cite{chiribella2010purification,coecke2011categorical,kent2012minkowski,zurek2009darwinism,vitanyi2001kolmogorov}. Cloning distinguishes between classical and quantum mechanics, being easy in one and impossible in the other. The same is true, for example, for programming \cite{nielsen1997programmable}, the universal-NOT gate \cite{buzek1999unot} and deleting \cite{no_deleting}.

First, we would like to motivate another such task: the superposition. Given some predetermined weights $\alpha,\beta\in\mathbb{R}_+$; the task is to output a state proportional to
\begin{equation}
    \alpha\ket{u}+\beta\ket{v}\label{eq:superpos_simple}
\end{equation}
from any two input states $\ket{u},\ket{v}\in\hat{\h}$, where $\h$ is a finite-dimensional Hilbert space and $\hat{\h}\subset\h$ denotes the subset of unit vectors\footnote{The precise definition will specify the relative phase between the $\alpha$ and $\beta$ terms, so no generality is lost by assuming $\alpha,\beta\in\mathbb{R}_+$.}. The superposition can be seen as a quantum version of addition. Graphically in \eqref{eq:superpos_simple}, the superposition corresponds to inserting $+$ between the two rescaled vectors. How hard is it to implement this $+$? This fundamental question is interesting in its own right.

Moreover, if easily implementable, $+$ could serve as an elementary gate for quantum computation. Consider Aaronson’s \emph{quantum state tree} model \cite{aaronson2004multilinear,Cai_2015} whose elementary operations are $+$ and the trivially implementable tensor product. In ref. \cite{aaronson2004multilinear} the purpose of the quantum state tree model was to capture the hardness of a state’s classical description, but suppose for a moment that $+$ was easily implementable. Then, the quantum state tree model could directly instruct a state's preparation process. This way to classically describe a state would also correspond to a new approach to physically preparing it.

The superposition question was studied previously. Algorithms were found for cases when the input states are restricted or partially known \cite{aaronson2004multilinear, doosti2017universal, oszmaniec2016creating,alvarez2015forbidden,li2017approximate,hu2016experimental,li2017experimentally}. For example, algorithms exist for inputs restricted to orthogonal pairs 
if these inputs are qubits \cite{doosti2017universal} or have known preparation circuits \cite{aaronson2004multilinear}\footnote{See the proof of Theorem 6 (OTree $\subseteq \Psi$P) in \cite{aaronson2004multilinear}.}. Another restriction of the inputs that makes an algorithm possible fixes their overlap with a reference state \cite{oszmaniec2016creating}. When the input states are unrestricted and completely unknown, Oszmaniec, Grudka, Horodecki, and W{\'o}jcik \cite{oszmaniec2016creating} showed that a quantum circuit cannot create an exact superposition from one copy of each input.

In this work, we define and study the superposition question when the input states $\ket{u}, \ket{v}\in\hat{\h}$ are completely unrestricted and unknown, $N$ copies of each input are available, and approximations are allowed. We find that in the most commonly considered models of quantum computation preparing such a superposition is impossible for any $N$; some inputs necessarily produce a large approximation error in the output. This lower bound on the error is tight. Fortunately, complicating the computational model overcomes the impossibility, as it allows employing state tomography in a useful way. Therefore, our strengthened superposition impossibility has new consequences. In addition to widening the operational distinction between the quantum and classical theories that follows already from the impossibility of ref. \cite{oszmaniec2016creating}, our results yield a new separation between different \emph{quantum} computational models, and a better understanding of the topological method developed in \cite{gavorova2024topological} and employed here. 

\subsection{Precise problem statement}

Correctly formulating the problem for completely unknown inputs has a subtlety. Oszmaniec et al. \cite{oszmaniec2016creating} observed that for unknown $\ket{u},\ket{v}\in\hat{\h}$, preparing \eqref{eq:superpos_simple} is trivially impossible. Suppose an equal-weight superposition protocol succeeds for $\ket{u}=\ket{0}$ and $\ket{v}=\ket{1}$, outputting $\ket{0}+\ket{1}$. If we change $\ket{1}$ to $-\ket{1}$ should we expect $\ket{0}-\ket{1}$? Distinguishing the global phase of a vector is unphysical, and thus trivially impossible for any protocol. Taking this into account, our superposition definition should again accept $\ket{0}+\ket{1}$ as the correct answer. It should accept outputs that ignore the global phases of $\ket{u},\ket{v}$, and instead depend only on the physically relevant representation of the inputs -- as density matrices. Instead of the unnormalized state in \eqref{eq:superpos_simple}, we can consider, for example,
\begin{equation}
    \Suv{+}=\alpha \ket{u} \left|\left\langle v|u \right\rangle\right|
    +
    \beta
    \ket{v}\!\left\langle v| u\right\rangle\text{,}\label{eq:sup_example_+}
\end{equation}
which ensures that $\Suv{+}\!\Suvbra{+}$ is a function of the density matrices $\ket{u}\!\bra{u}$, $\ket{v}\!\bra{v}$. The renormalization of $\ket{s_+{\scriptstyle(u,v)}}$ is well defined on all $\ket{u}$, $\ket{v}$ nonorthogonal, which is a dense subset of $\hat{\h}\times\hat{\h}$. We generalize \eqref{eq:sup_example_+} for our final definition.

\begin{samepage}
\begin{mydef}\label{def_sup}
A superposition of $\ket{u},\ket{v}\in\hat{\h}$ is a state proportional to
\begin{equation}
    \Suv{c}=
    \alpha \ket{u} \left|c_{uv}\right|
    +
    \beta
    \ket{v} c_{uv}\text{,}
\end{equation}
for some complex function $c_{uv}=c(\ket{u},\ket{v})$ such that $S:=c^{-1}(\mathbb{C}\setminus\{0\})$ is a dense subset of $\hat{\h}\times\hat{\h}$.
\end{mydef}
\end{samepage}

Compared to Oszmaniec et al. \cite{oszmaniec2016creating}\footnote{The definition of Oszmaniec et al. is equivalent to the following: $\ket{s}$ is a superposition of $\ket{u}$ and $\ket{v}$ with coefficients $\alpha$, $\beta$ iff $\ket{s}\!\bra{s}= N(\alpha\ket{u}+e^{i\phi}\beta\ket{v})(\alpha^*\bra{u}+e^{-i\phi}\beta^*\bra{v})$ for some real $N, \phi$. Restricting $\phi$ to be a deterministic function of $\ket{u}$ and $\ket{v}$ falls within our definition with $e^{i\phi}=c_{uv}/|c_{uv}|$. This $\phi$-function can be undefined on $\ket{u}, \ket{v}$ outside the dense subset.}, this article formalizes the superposition differently: It defines the superposition of $\ket{u}$ and $\ket{v}$ as a \emph{deterministic function} of $\ket{u}$ and $\ket{v}$, but it also considers \emph{sets} of superpositions (each superposition coming from a different $c$-function). This allows us to postpone the (non)determinism nuance to our later discussion of computational models.  

Definition \ref{def_sup} includes \eqref{eq:sup_example_+} and other choices of the form $c_{uv}=\bra{v}M_{uv}\ket{u}$. If the matrix $M_{uv}$ depends only on the inputs' density matrices, then so does the superposition. For $c_{uv}=\left\langle v|j\right\rangle\! \left\langle i|u\right\rangle $ with $i,j\in\{0,1,\dots, \dim\h-1\}$, the resulting superposition is (on the corresponding dense subset) proportional to
\begin{equation}
    \alpha \ket{u}\frac{\left\langle u|i\right\rangle}{\left|\left\langle u|i\right\rangle\right|}
    +
    \beta \ket{v}
    \frac{\left\langle v|j\right\rangle}{\left|\left\langle v|j\right\rangle\right|}
    =
    \alpha \vv_i(\ket{u}\!\bra{u})
    +
    \beta \vv_j(\ket{v}\!\bra{v})\text{.}\label{eq:superpos_example2}
\end{equation}
In this particular choice of superposition we can treat each input separately; each input has its own function $\ket{u}\!\bra{u}$ has $\vv_i$, $\ket{v}\!\bra{v}$ has $\vv_j$) that maps the rank-$1$ density matrix to a corresponding\footnote{A vector $\ket{w}\in\hat{\h}$ corresponds to $\rho\in\D_\text{pure}(\h)$ iff $\ket{w}\!\bra{w}=\rho$.} vector. The function $\vv_i$ chooses the vector by fixing its $i$-th entry, $\bra{i}\vv_i(\rho)$, to be real positive.\\

In this work we study superposition protocols, considering their error and sample complexity. The error is quantified by the trace-norm distance from the target output, the superposition. We are interested in the worst-case performance of the protocols; the greatest error obtained for any valid input $(\ket{u},\ket{v})\in S$, i.e., any input such that the target output is renormalizable. Demanding worst-case guarantees is common in complexity theory and the reason is practical. An unwarranted error on some inputs is undesirable, especially if the solution is to be used as an elementary operation, as is the case for $+$. Worst-case error bounds are satisfied, for example, by state tomography protocols \cite{sugiyama2013precision,christandl2012reliable}, which also set the expectation for our second performance measure: sample complexity, or the number of input copies used by the protocol. The following superposition protocol seems possible with $N=\exp(n)$ samples of each $n$-qubit input ($n=\lceil\log_2\dim\h\rceil$): Perform state tomography to get an approximate classical description of each of the two input states, then classically describe their superposition, then build a new circuit that transforms the all-zero input into this described state\footnote{\label{footnote4}This is always possible by the universality of quantum circuits.}. We might, therefore, expect the optimal protocol’s sample complexity to exceed $1$, due to the previous no-go result \cite{oszmaniec2016creating}, and not to exceed $\exp(n)$, due to state tomography. Our results below defy this expectation, revealing the importance of distinguishing between different computational models. 

\subsection{Overview of the results in different computational models}
The answer to our superposition question depends on which model of computation we consider: the trace-preserving model, the postselection (probabilistic) model, or the random model. For each we give a definition, an example protocol, and our result.

\emph{Trace-preserving model.} The trace-preserving model with sample complexity $N$ is a completely positive trace-preserving map (CPTP) applied to $\rho_\text{in}^{\otimes N}$, for example, to $(\ket{u}\!\bra{u}\otimes\ket{v}\!\bra{v})^{\otimes N}$. The output approximates a given target output, for example, a given superposition $\alpha \ket{u} \left|c_{uv}\right|
    +
    \beta
    \ket{v} c_{uv}$ determined by a particular $c$-function.

\begin{figure}
\begin{minipage}[c]{0.4\linewidth}
    \centering
    \subfloat[]{
         \hspace{0pt}\includegraphics[scale=0.8,valign=b]{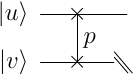}
         \label{fig:p-swap}
     }
     \hspace{-10pt}
     \subfloat[]{
         \includegraphics[scale=0.8,valign=b]{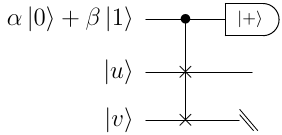}
         \label{fig:c-swap}
     }
        \caption{(a) The probabilistic swap protocol swaps the two registers with probability $p=\sin^2\theta$ and does not swap them with probability $1-p=\cos^2\theta$. Then, the second register is traced out. (b) The controlled swap protocol has an additional control qubit to control the controlled swap and a control measurement to postselect the $\ket{+}=(\ket{0}+\ket{1})/\sqrt{2}$ outcome.}
        \label{fig:two graphs}
\end{minipage}
\hfill
\begin{minipage}[c]{0.58\linewidth}
    \centering
        \includegraphics[scale=0.8,valign=c]{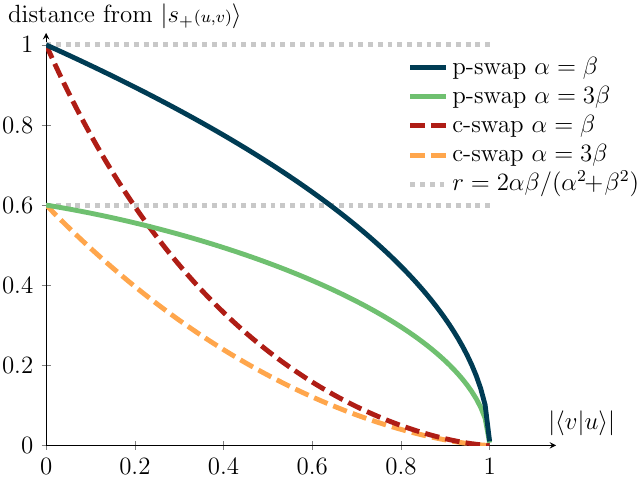}\hspace{-10pt}
        \caption{The error as a function of the inputs for different protocols and $\alpha :\beta$ ratios. The controlled swap protocol (c-swap) outperforms the probabilistic swap protocol (p-swap) on almost all inputs, but the worst-case error of both is $r$.}\label{fig:protocols_plot}
\end{minipage}
\end{figure}

Next, we give an example protocol in the trace-preserving model for the target superposition $\Suv{+}$. The protocol $\Psi$ accepts $N=1$ samples of each state and outputs the mixed state $\cos^2\theta\ket{u}\!\bra{u}+\sin^2\theta\ket{v}\!\bra{v}$, where $\cos\theta=\alpha/\sqrt{\alpha^2+\beta^2}$. A probabilistic swap (\Cref{fig:p-swap}) implements this $\Psi$. For $\ket{u},\ket{v}\in S$ (nonorthogonal pairs from $\hat{\h}$) it satisfies
\begin{eqnarray}
    &\left|\left|\Psi(\ket{u}\!\bra{u}\otimes\ket{v}\!\bra{v})-\frac{\Suv{+}\!\Suvbra{+}}{\left\langle s_{\operatorname{+}}{\scriptstyle(u,v)}|s_{\operatorname{+}}{\scriptstyle(u,v)}\right\rangle}\right|\right|_{\tr} =r\sqrt{\frac{1-r|\left\langle v | u\right\rangle|}{1+r|\left\langle v | u\right\rangle|}}\sqrt{1-|\left\langle v | u\right\rangle|^2}\text{,}&\\
    &\text{where }r=\sin(2\theta)=\frac{2\alpha\beta}{\alpha^2+\beta^2}\text{,}&
\end{eqnarray}
and $\left|\left|O\right|\right|_{\tr}:=\tr\sqrt{OO^\dagger}$ so the maximal trace distance of two density states is $2$. Above, the output state is the furthest from the desired state as $\left\langle v | u\right\rangle\to 0$, when the trace distance equals $r$ (see \Cref{fig:protocols_plot}). We say the worst-case error of this protocol is $r$. Our first result discusses the worst-case performance of \emph{any} protocol in the trace-preserving model.

\emph{Result 1.}\label{res1} For any choice of target superposition and any sample complexity $N$, no trace-preserving protocol achieves a worst-case error smaller than $r$. 
In light of the above protocol, this impossibility is tight.

\emph{Postselection model.} The postselection model with sample complexity $N$ corresponds to a completely positive trace-nonvanishing map (CPTNV) applied to $\rho_\text{in}^{\otimes N}$. Only trace-nonincreasing maps are physical. Trace-nonvanishing means that for any input state, the map outputs a nonzero\footnote{The nonzero requirement, also appearing in the definition of the complexity class PostBQP \protect{\cite{aaronson_PostBQP}}, allows the model to remain worst-case.} positive semidefinite matrix so that it can be renormalized and then compared to a given target output. The postselection model generalizes the trace-preserving model.

Our example protocol $\Phi$ is again for the $\Suv{+}$ superposition and $N=1$. It initializes a control qubit to $\cos\theta\ket{0}+\sin\theta\ket{1}$, then controlled-swaps $\ket{u}$ and $\ket{v}$, then projects the control onto the $\ket{+}$ state (\Cref{fig:c-swap}). For nonorthogonal $\ket{u},\ket{v}$ this CPTNV $\Phi$ has the error
\begin{align}
    \left|\left|\frac{\Phi(\ket{u}\!\bra{u}\otimes\ket{v}\!\bra{v})}{\tr \left[\Phi(\ket{u}\!\bra{u}\otimes\ket{v}\!\bra{v})\right]}-\frac{\Suv{+}\!\Suvbra{+}}{\left\langle s_{\operatorname{+}}{\scriptstyle(u,v)}|s_{\operatorname{+}}{\scriptstyle(u,v)}\right\rangle}\right|\right|_{\tr} & =r\sqrt{\frac{1-r|\left\langle v | u\right\rangle|}{1+r|\left\langle v | u\right\rangle|}}\sqrt{1-|\left\langle v | u\right\rangle|^2}\frac{1-|\left\langle v | u\right\rangle|}{1+r|\left\langle v | u\right\rangle|^2}\text{,}
\end{align}
an improvement for $|\left\langle v | u\right\rangle|\in(0,1)$, but not in the worst case, when the error again equals $r$ (\Cref{fig:protocols_plot}). Indeed, our matching lower bound extends to the postselection model.

\emph{Result 1*.}\label{res11} The worst-case trace distance $r$ is optimal for any $N$ and any choice of superposition also in the postselection model.

\emph{Random model.} The random model with sample complexity $N$ corresponds to a quantum instrument $\{\Psi_i\}_{i\in I}$ (an indexed set of CP maps that sum up to a CPTP) applied to $\rho_{\text{in}}^{\otimes N}$. We obtain the label $i$ with probability $\tr\left[\Psi_i(\rho_{\text{in}}^{\otimes N})\right]$, in which case the algorithm’s quantum output $\Psi_i(\rho_{\text{in}}^{\otimes N})$ approximates (up to the renormalizations) the $i$-th target output, for example, the $i$-th superposition $\alpha\ket{u}|c_{i\,uv}|+\beta\ket{v}c_{i\,uv}$. Thus, in this model we request an indexed \emph{set} of $|I|$ targets. Which target’s approximation is the output of a particular run of the protocol is random, but revealed by the classical outcome $i$. This model includes the trace-preserving model (obtained by setting $|I|=1$).

An example protocol in the random model is the protocol for qubit inputs $\ket{u},\ket{v}\in\mathbb{C}^2$ of Doosti, Kianvash and Karimipour \cite{doosti2017universal}, which has $N=1$ and zero error if the inputs are restricted to orthogonal states, $\left\langle v | u \right\rangle = 0$. An example protocol with unrestricted inputs is the state tomography subroutine suggested above as a building block of a possible superposition protocol. The na\"ive suggestion is to run a state tomography for both $\ket{u}$ and $\ket{v}$, to use the resulting classical estimates to calculate a classical description of a state close to 
\begin{equation}
    \left(\alpha \ket{u}+\beta\ket{v}\right)\left(\alpha \bra{u}+\beta\bra{v}\right)\label{expr:superpos_state}\text{,}
\end{equation}
and then to build a new quantum circuit that outputs a state that is arbitrarily close to that classical description (see \cref{footnote4}). To estimate \eqref{expr:superpos_state}, classical estimates of the \emph{density matrices} $\ket{u}\!\bra{u}$, $\ket{v}\!\bra{v}$ do not suffice, because \cref{expr:superpos_state} contains cross-terms. Thus, we ask for \emph{vector tomography}, a state tomography that outputs a classical estimate of a vector $\vv(\ket{u}\!\bra{u})\in\hat{\h}$.

\emph{Result 2.}\label{res2} Vector tomography is not possible in the trace-preserving and postselection models.

\emph{Result 3.}\label{res3} Vector tomography is possible in the random model.

In the above definitions, randomness enters protocols in two crucially distinct ways. A postselection protocol randomly (sometimes) generates an output, but the output itself is deterministic. A random protocol deterministically (always) generates an output, but the output itself is random.
The results reveal that this distinction is important;
the random model is strictly stronger than the postselection model. As a corollary of Result \hyperref[res3]{3}, vector tomography gives a superposition protocol in the random model, with sample complexity exponential in $n=\lceil\log_2\dim\h\rceil$. The optimality of this protocol remains open.

\subsection{Topology as a tool to prove no-go results}\label{sec_method}

Results \hyperref[res1]{1}, \hyperref[res11]{1*}, and \hyperref[res2]{2} are all no-go results. They use the method developed in \cite{gavorova2024topological}, which combines two observations:
\begin{enumerate}
    \item any quantum circuit corresponds to a continuous function,
    \item the task requires implementing a function that cannot be continuous by topological arguments.
\end{enumerate} 
In ref. \cite{gavorova2024topological}, the function’s input is a unitary oracle. Here, it is a pair of states. Specifically, the first observation corresponds to:

\begin{myprop}\label{prop_cont} If a postselection algorithm takes as the input any number of copies of $\ket{u}\!\bra{u}$ and $\ket{v}\!\bra{v}$, then its output is a continuous function of $\ket{u}\!\bra{u}, \ket{v}\!\bra{v}$.
\end{myprop}

\noindent This is a weaker version of the well-known observation that any quantum circuit corresponds to a polynomial \cite{lbounds_polynomials}.
A lower bound on the polynomial degree of the required function implies a lower bound on the sample complexity of any implementing algorithm. Continuity is more general. Disproving the continuity of the required function excludes algorithms of \emph{any} sample complexity. Can we relate our problem to some continuous functions forbidden by topology? Comparing to \cite{gavorova2024topological} we suggest two functions.

The first function is a \emph{section}. Topology forbids a continuous function, also called \emph{section}, that maps a pure density operator to a corresponding vector. In other words, no continuous $\vv:\D_\text{pure}(\h)\to\hat{\h}$ satisfies $\vv(\rho)\vv(\rho)^\dagger=\rho$, i.e., right-inverts the outer product map\footnote{The outer product map $\pi:\hat{\h}\to\D_\text{pure}(\h)$, $\ket{u}\mapsto\ket{u}\!\bra{u}$ is the familiar Hopf fibration. When the input $\ket{u}=(u_0, u_1)^T\in\hat{\mathbb{C}^2}$ is a qubit, the Hopf fibration $\pi(\ket{u})=\ket{u}\!\bra{u}=(I+r\cdot\vec{\sigma})/2$ with $r=(\operatorname{Re}[2u_0^*u_1], \operatorname{Im}[2u_0^*u_1], |u_0|^2-|u_1|^2)^T$ lets us represent the qubit on the Bloch sphere, because $\D_\text{pure}(\mathbb{C}^2)$ is homeomorphic to the $2$-sphere ($r\in\Sp^2$).}. This is analogous to no-section for unitaries \cite{gavorova2024topological}, where no continuous function maps a unitary superoperator $\mathcal{U}(\cdot)=U\cdot U^\dagger$ to a “canonical” unitary $\mathcal{U}\mapsto \widetilde{U}$, such that $\widetilde{U}\cdot\widetilde{U}^\dagger=\mathcal{U}$. For states, the nonexistence of a section yields a no-go for the exact superposition when $\dim\h\geq 3$ (see \Cref{sec_alt_proof} for the proofs). To deal with approximations, we need to “quantify” the discontinuity in any $\vv$. Intuitively, the desired maps $\vv: \ket{u}\!\bra{u}\mapsto\ket{\widetilde{u}}$ and $U\cdot U^\dagger\mapsto \widetilde{U}$ have something in common; as opposed to the domains, the ranges distinguish the global phase $\lambda\in \operatorname{U}(1)\approx\Sp^1$, where we note the homeomorphism between $\operatorname{U}(1)$ and the circle $\Sp^1$ (implicit from now on). The next function extracts this $\Sp^1$ structure to contradict continuity, exploiting the nontrivial topology of the circle $\Sp^1$: different numbers of loops on $\Sp^1$ cannot be continuously contracted to each other.

The second function is a \emph{continuous homogeneous function into $\Sp^1$}. For $m\in\mathbb{Z}$, a function $f$ is $m$-homogeneous, if for all scalars $\lambda$ and inputs $x$, it satisfies $f(\lambda x)=\lambda^m f(x)$. If the domain of $f$ is $\hat{\h}$, we only consider $|\lambda|=1$. For example, on $\hat{\h}$ the outer product map $\ket{u}\mapsto\ket{u}\!\bra{u}$ is $0$-homogeneous. We find excluding \emph{continuous homogeneous functions into $\Sp^1$} more useful for quantifying approximations. Specifically, denoting by $\hat{\mathbb{C}^2}$ the space of norm-$1$ vectors in $\mathbb{C}^2$, our no-go results rely on the following lemma:
    
\begin{mylemma}\label{lem_2homog} Let $m\neq 0$. There is no continuous $m$-homogeneous function $\hat{\mathbb{C}^2}\to \Sp^1$.
\end{mylemma}

Next, in \Cref{sec:superpos_nogo} we prove  Results \hyperref[res1]{1} and \hyperref[res11]{1*} using \Cref{lem_2homog}, then we prove \Cref{lem_2homog}. In \Cref{sec:tomography} we prove Results \hyperref[res2]{2} and \hyperref[res3]{3}. In \Cref{sec:discussion} we discuss the significance of our results for setting quantum mechanics apart from classical, for harnessing the power of random outputs, and for understanding the applicability of the topological method.

\section{The superposition impossibility}\label{sec:superpos_nogo}

Results \hyperref[res1]{1} and \hyperref[res11]{1*} are the superposition no-gos in the trace-preserving and the postselection models. Since the trace-preserving model is a special case of the postselection model, we focus on the latter. Denote by $\D_+(\h)$ the set of positive semidefinite and nonzero linear operators from $\h$ to itself. By definition, the outputs of a postselection algorithm are in $\D_+(\h)$.

\begin{mytheorem}\label{thm_superpos}
 Let $\alpha,\beta\in\mathbb{R}_+$. Let any function $c:\hat{\h}\times\hat{\h}\to\mathbb{C}$, such that $S:=c^{-1}(\mathbb{C}\setminus\{0\})$ is dense in $\hat{\h}\times\hat{\h}$, determine the superposition $\Suv{c}=\alpha \left|c(u,v)\right|\ket{u}+\beta\, c(u,v)\ket{v}$. There is no continuous function $\mathcal{A}:\mathcal{D}_\text{pure}(\h)\times \mathcal{D}_\text{pure}(\h)\to \mathcal{D}_+(\h)$ such that 
\begin{align*}
\sup_{(\ket{u},\ket{v})\in S}\left|\left|\frac{\mathcal{A}(\ket{u}\!\bra{u},\ket{v}\!\bra{v})}{\operatorname{tr}\left[\mathcal{A}(\ket{u}\!\bra{u},\ket{v}\!\bra{v})\right]} - 
    \frac{\Suv{c}\!\Suvbra{c}}{\left\langle s_{\operatorname{c}}{\scriptstyle(u,v)}|s_{\operatorname{c}}{\scriptstyle(u,v)}\right\rangle}\right|\right|_{\operatorname{tr}} &< \frac{2\alpha\beta}{\alpha^2+\beta^2}\text{.}
\end{align*}
\end{mytheorem}

\noindent Combining with \Cref{prop_cont} gives the no-gos.

\begin{proof}
Assume towards contradiction that a function $\mathcal{A}$ of Theorem \ref{thm_superpos} exists. We will use it to build a function excluded by \Cref{lem_2homog}. Let $u=(u_0, u_1)^T\in \hat{\mathbb{C}^2}$. We define vectors in $\hat{\h}$ that are homogeneous functions of $u$ by adding trailing zeros, namely
\begin{eqnarray*}
    \ket{u}:=\begin{pmatrix}
        u_0\\ u_1\\0\\ \vdots
    \end{pmatrix}
    \text{, }&
    \ket{u^{\pperp}}:=\begin{pmatrix}
        u^*_1\\ -u^*_0\\0\\ \vdots
    \end{pmatrix}\text{,}&
    \text{and }\bra{u^{\pperp}}=
    \begin{pmatrix}
        u_1 & -u_0 & 0 & \dots
    \end{pmatrix}
\end{eqnarray*}
are respectively $1$-homogeneous, $(-1)$-homogeneous, and $1$-homogeneous. 
Then the following function 
\begin{equation}
    g(u):=\bra{u^{\pperp}}\frac{\mathcal{A}\big(\ket{u}\!\bra{u},\ket{u^{\pperp}}\!\bra{u^{\pperp}}\big)}{\operatorname{tr}\left[\mathcal{A}(\ket{u}\!\bra{u},\ket{u^{\pperp}}\!\bra{u^{\pperp}})\right]}\ket{u}\text{,}\label{eq_g}
\end{equation}
is a continuous, $2$-homogeneous function $\hat{\mathbb{C}^2}\to\mathbb{C}$. Since the matrix inside the contractive map $\bra{u^{\pperp}}\cdot\ket{u}$ is positive semidefinite and trace-$1$, we have $|g(u)|\leq 1$. In the remaining part of the proof we show that $|g(u)|>0$ for all $u\in \hat{\mathbb{C}^2}$, which allows us to define $\hat{g}(u):=g(u)/|g(u)|$. The map $\hat{g}$ is also continuous and $2$-homogeneous and, moreover, maps into $\Sp^1$. The existence of such $\hat{g}$ contradicts \Cref{lem_2homog}. 

To prove $|g(u)|>0$ we first observe that even if $c(\ket{u}, \ket{u^{\pperp}})=0$, the set $S:=c^{-1}(\mathbb{C}\setminus\{0\})$ is dense in $\hat{\h}\times\hat{\h}$, so inside $S$ there exists a converging subsequence $(\ket{u_n}, \ket{u^{\pperp}_n})$ such that
\begin{align*}
    &\lim_{n\to\infty}\ket{u_n}=\ket{u}\\
    &\lim_{n\to\infty}\ket{u^{\pperp}_n}=\ket{u^{\pperp}}\text{.}
\end{align*}
We choose this notation because $\lim_{n\to\infty}\left\langle u^{\pperp}_n|u_n\right\rangle = 0$, but note that for a given $n$ the inner product $\left\langle u^{\pperp}_n|u_n\right\rangle$ might be nonzero. Our assumption towards contradiction implies
\begin{align}
\frac{2\alpha\beta}{\alpha^2+\beta^2}>&
    \lim_{n\to\infty} \bigg|\bigg|\frac{\mathcal{A}(\ket{u_n}\!\bra{u_n},\ket{u_n^{\pperp}}\!\bra{u_n^{\pperp}})}{\operatorname{tr}\left[\mathcal{A}(\ket{u_n}\!\bra{u_n},\ket{u_n^{\pperp}}\!\bra{u_n^{\pperp}})\right]} - 
    \frac{\ket{s_{\operatorname{c}}{\scriptstyle(u_n,u_n^{\pperp})}}\!\bra{s_{\operatorname{c}}{\scriptstyle(u_n,u_n^{\pperp})}}}{\left\langle s_{\operatorname{c}}{\scriptstyle(u_n,u_n^{\pperp})}|s_{\operatorname{c}}{\scriptstyle(u_n,u_n^{\pperp})}\right\rangle}\bigg|\bigg|_{\operatorname{tr}}\text{.}\label{eq_lim}
\end{align}

Denote by $X_n$ the entire expression inside the trace norm above. $X_n$ is Hermitian. For any positive trace-nonincreasing superoperator $\Psi$, $\left|\left|X_n\right|\right|_{\tr}\geq \left|\left|\Psi(X_n)\right|\right|_{\tr}$. Choose $\Psi(X_n) = \ket{0}\!\bra{u^{\pperp}}X_n\ket{u}\!\bra{0}+\ket{1}\!\bra{u}X_n\ket{u^{\pperp}}\!\bra{1}$. This $\Psi$ is (completely) positive trace-nonincreasing, because it corresponds to the conjugation by the operator $\ket{00}\!\bra{u^{\pperp}}+\ket{11}\!\bra{u}$ followed by a partial trace. Since $X_n$ is Hermitian, we get $\left|\left|X_n\right|\right|_{\tr}\geq \left|\left|\Psi(X_n)\right|\right|_{\tr}=2\left|\bra{u^{\pperp}}X_n\ket{u}\right|$. Substituting this into \eqref{eq_lim}, cancelling the factors of $2$, and using the triangle inequality, we get
\begin{align}
\frac{\alpha\beta}{\alpha^2+\beta^2}>&
    \lim_{n\to\infty}\bigg[-\left|\bra{u^{\pperp}}\frac{\mathcal{A}(\ket{u_n}\!\bra{u_n},\ket{u_n^{\pperp}}\!\bra{u_n^{\pperp}})}{\operatorname{tr}\left[\mathcal{A}(\ket{u_n}\!\bra{u_n},\ket{u_n^{\pperp}}\!\bra{u_n^{\pperp}})\right]}\ket{u}\right|
    + 
    \left|\frac{\left\langle u^{\pperp} |s_{\operatorname{c}}{\scriptstyle(u_n,u_n^{\pperp})}\right\rangle\left\langle s_{\operatorname{c}}{\scriptstyle(u_n,u_n^{\pperp})}|u\right\rangle}{\left\langle s_{\operatorname{c}}{\scriptstyle(u_n,u_n^{\pperp})}|s_{\operatorname{c}}{\scriptstyle(u_n,u_n^{\pperp})}\right\rangle}\right|\bigg]\label{eq_lim2}
\end{align}
Inside the limit, the first term converges to $|g(u)|\in[0,1]$. Denoting $c(\ket{u_n},\ket{u_n^{\pperp}})=:c_n=|c_n|e^{i\gamma_n}$ the second term converges to
\begin{align*}
    \frac{|c_n|^2\left|\bra{u^{\pperp}}\left(\alpha \left|u_n\right\rangle+e^{i\gamma_n}\beta\left|u^{\pperp}_n\right\rangle\right)\left(\alpha \left\langle u_n\right|+e^{-i\gamma_n}\beta\left\langle u^{\pperp}_n\right|\right)\ket{u}\right|}{|c_n|^2\left(\alpha^2+\beta^2+2\alpha\beta\operatorname{Re}\left[e^{i\gamma_n}\left\langle u^{\pperp}_n|u_n\right\rangle\right]
    \right)} 
    &\xrightarrow{n\to\infty}\frac{|e^{i\gamma_n}\beta\alpha|}{\alpha^2+\beta^2}\text{.}
\end{align*}
Inequality \eqref{eq_lim2} becomes
\begin{equation*}
    \frac{\alpha\beta}{\alpha^2+\beta^2}>-\left|g(u)\right| + 
    \frac{|e^{i\gamma_n}\beta\alpha|}{\alpha^2+\beta^2}
    \text{,}
\end{equation*}
which implies $|g(u)|>0$, the missing step to establish the contradiction with \Cref{lem_2homog}.
\end{proof}

\begin{proof}[Proof of \Cref{lem_2homog}]
We parametrize paths on $\hat{\mathbb{C}^2}$ by the parameter $t\in[0,1]$.  
The following paths are loops on $\hat{\mathbb{C}^2}$ based at a fixed point $\ket{u_b}\in\hat{\mathbb{C}^2}$:
\begin{eqnarray*}
    f(t):=\ket{u_b}\text{, } &&
    f'(t):=e^{i2\pi t}\ket{u_b}\text{.}
\end{eqnarray*}
Namely, $f$ is constant and $f'$ loops in its global phase. The paths are continuously contractible to each other, or homotopic, because on $\hat{\mathbb{C}^2}\approx \Sp^3$ any loop is continuously shrinkable to a point -- we say that $\Sp^3$ is simply connected. See Fig. \ref{fig_sphere} for an intuition.

Assume towards contradiction that there is a function $g:\hat{\mathbb{C}^2}\to \Sp^1$ which is continuous and $m$-homogeneous for $m\neq 0$, $m\in\mathbb{Z}$. Composing $g$ onto our loops on $\hat{\mathbb{C}^2}$ gives us loops on the circle, $\Sp^1$: 
\begin{eqnarray*}
    (g\circ f)(t)&=&g(\ket{u_b})=:\lambda_b\\
    (g\circ f')(t)&=&g(e^{i2\pi t}\ket{u_b})=e^{i 2\pi t m}\lambda_b\text{.}
\end{eqnarray*}
The function $g\circ f$ is constant -- it corresponds to zero loops on $\Sp^1$, while $g\circ f'$ by the $m$-homogeneity of $g$ makes $|m|$ loops on $\Sp^1$. By the continuity of $g$ the loops $g\circ f$ and $g\circ f'$ must also be continuously deformable to each other. But on the circle $\Sp^1$ this is impossible for different numbers of loops. 
\end{proof}

\begin{figure}[b]
    \centering
    \includegraphics{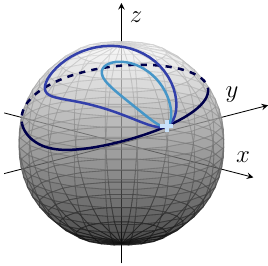}
    \vspace{1.5em}
    \caption{Any closed loop on $\Sp^2$ is continuously contractible to a point. The same is true for $\Sp^3$ which is homeomorphic to $\hat{\mathbb{C}^2}$. We say that these spaces are simply connected.}
    \label{fig_sphere}
    \vspace{1em}
\end{figure}

\section{State tomography for superposition}\label{sec:tomography}
Here we prove Results \hyperref[res2]{2} and \hyperref[res3]{3} about vector tomography. Instead of outputting a classical description of a pure density state in $\D_\text{pure}(\h)$, vector tomography outputs a classical description of a vector in $\hat{\h}$. However, different vectors in $\hat{\h}$ match the same rank-1 matrix in $\D_\text{pure}(\h)$. Choosing one (canonical) vector for each matrix is the purpose of the function $\vv: T\to\hat{\h}$, where $T$ is dense in $\D_\text{pure}(\h)$. The requirement $\vv(\rho)\vv(\rho)^\dagger=\rho$ ensures that the vector matches the matrix. An example is
\begin{equation}
    \vv_0(\rho) =  \frac{\rho\ket{0}}{\sqrt{\bra{0}\rho\ket{0}}}\text{,}\label{eq:v_example}
\end{equation}
which accepts inputs with a nonzero first matrix element, $T=\{\rho \in \D_\text{pure}(\h);\bra{0}\rho\ket{0}\neq 0\}$, and chooses a vector with the first element $\bra{0}\vv_0(\rho)$ real positive. 

Vector tomography should satisfy a precision guarantee similar to that of state tomography. State tomography protocols \cite{sugiyama2013precision,christandl2012reliable} satisfy a guarantee in terms of trace norm $||\cdot||_{\tr}$; given $N$ samples of $\rho$ they output a classical description $x$ from a set of values $ T_N$ according to a probability distribution $p_{N,\rho}(x)$ such that
\begin{equation}
    Pr_{x\sim p_{N,\rho}(\cdot)}\left[\left|\left|x-\rho\right|\right|_{\tr}\leq \epsilon_N^{\tr}\right]\geq 1-\delta_N^{\tr}\text{,}\label{eq:tom_tr}
\end{equation}
where $\lim_{N\to\infty}\epsilon_N^{\tr}=0$ and  $\lim_{N\to\infty}\delta_N^{\tr}=0$. The classical output $x$ is a result of quantum measurements and a classical postprocessing, which can ensure $x\in \D_\text{pure}(\h)$ (as in our case $\rho\in \D_\text{pure}(\h)$). The procedure is in principle convertible to a fully quantum algorithm with a deferred measurement, so in \eqref{eq:tom_tr} we can take $x\in T_N\subset \D_\text{pure}(\h)$. This equation-\eqref{eq:tom_tr} characterization of state tomography inspires our following, equation-\eqref{eq_sttom} characterization of vector tomography.

\begin{mydef}[Vector tomography]\label{def_sttom}
A sequence of $ T_N$-algorithms is a \emph{vector tomography} if  $T:=\bigcup_{N=0}^{\infty}  T_N$ is dense in $\D_\text{pure}(\h)$, if there exists $\vv: T\to \hat{\h}$ such that $\vv(\rho)\vv(\rho)^\dagger=\rho$, and if the $N$-th algorithm works as follows: Its quantum part on input $\rho^{\otimes N}$ where $\rho\in \mathcal{D}_\text{pure}(\h)$ outputs a measurement outcome $x\in T_N$ with probability $p_{N,\rho}(x)$, its classical part then calculates $\vv(x)$, and its probability distribution $p_{N,\rho}$ for all $\rho\in T$ satisfies
\begin{equation}
    Pr_{x\sim p_{N,\rho}(\cdot)}\left[\left|\left|\vv(x)-\vv(\rho)\right|\right|_2\leq \epsilon_N\right]\geq 1-\delta_N\text{,}\label{eq_sttom}
\end{equation}
with $\lim_{N\to\infty}\epsilon_N=0$ and $\lim_{N\to\infty}\delta_N=0$.
\end{mydef}

Vector tomography of \Cref{def_sttom} gives a trace-preserving protocol. The $N$-th algorithm, including calculating $\vv(x)$ and writing it onto a classical register, can be made fully quantum, with a deferred measurement. Just before the deferred measurement, the reduced state on that classical register is a probabilistic mixture of basis states, each basis state indexed by $\vv(x)\in \vv( T_N)\subset \hat{\h}$:
    
    \begin{equation}
        \Psi(\rho^{\otimes N}):=\sum_{x\in T_N}p_{N,\rho}(x)\ket{e_{\vv(x)}}\!\bra{e_{\vv(x)}}\text{,}
    \end{equation}
where $\Psi$ is the $N$-th algorithm's CPTP. We measure the output’s closeness to the target $\vv(\rho)$ by using the linear map $h:e_i e_i^T\mapsto i$ from diagonal matrices to $\vv( T_N)$. The triangle inequality and condition \eqref{eq_sttom} imply that for all $\rho\in T$
    \begin{align}
        \left|\left|
    h\left(\Psi(\rho^{\otimes N})\right)-\vv(\rho)\right|\right|_2
    &=\left|\left|\sum_{x\in T_N}p_{N,\rho}(x)
    \vv(x)-\vv(\rho)\right|\right|_2\nonumber\\
    &\leq\sum_{x\in T_N}p_{N,\rho}(x)\left|\left|
    \vv(x)-\vv(\rho)\right|\right|_2\nonumber\\
    &
    \leq\sum_{x\text{ good}}p_{N,\rho}(x)\epsilon_N + \sum_{x\text{ bad}}p_{N,\rho}(x)\times 2\nonumber\\
    &\leq \epsilon_N + 2\delta_N\text{.}\label{eq_fN_close}
    \end{align}
In this sense, \Cref{def_sttom} implies a trace-preserving protocol. Since such a protocol would contradict the superposition impossibility, it cannot exist. Next, we show this directly.

\begin{mytheorem}
    The trace-preserving protocol $\Psi$ is impossible. Consequently, vector tomography of \Cref{def_sttom} is impossible.
\end{mytheorem}
\begin{proof}
\begin{figure}
    \centering
    \vspace{-10pt}
    \includegraphics[scale=0.8]
    {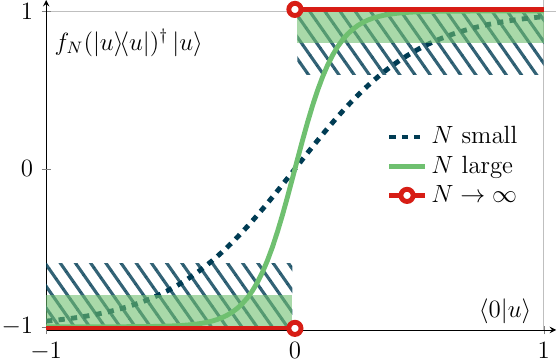}
    \caption{The contradiction illustrated on $\vv=\lim_{N\to\infty}f_N$ of \cref{eq:v_example}. This $\vv=\vv_0$ for a given matrix $\rho\in S\subset \D_\text{pure}(\h)$ outputs that corresponding vector whose first entry is real positive, $\left\langle 0|\vv_0(\rho)\right\rangle > 0$. The contradiction: functions $f_N$ should be both continuous and $(\epsilon_N+2\delta_N)$-close to $\vv$ (i.e., confined to the marked regions).}
    \label{fig:tomography_plot}
\end{figure}    
    The output state $\Psi(\rho^{\otimes N})$ is a continuous function of $\rho\in\D_\text{pure}(\h)$. Since $h$ is linear, the function
    \begin{equation}
        f_N(\rho):=h\circ\Psi(\rho^{\otimes N})=\mathbb{E}_{p_{N,\rho}}\left[\vv(x)\right]
    \end{equation}
    is also continuous in $\rho\in\D_\text{pure}(\h)$. Moreover, by \eqref{eq_fN_close} it is close to $\vv(\rho)$ for all $\rho\in T$. The continuity of $f_N(\rho)$ and its closeness to $\vv(\rho)$ lead to a contradiction illustrated in \Cref{fig:tomography_plot} and proven next. The function $g_N(\ket{u}):=f_N(\ket{u}\!\bra{u})^\dagger\ket{u}$ is continuous and $1$-homogeneous in $\ket{u}$. From \eqref{eq_fN_close} we get that whenever $\ket{u}\!\bra{u}\in T$
    \begin{align*}
        \left|g_N(\ket{u})\right|&\geq \left|\vv(\ket{u}\!\bra{u})\ket{u}\right|-\left|f_N(\ket{u}\!\bra{u})\ket{u}-\vv(\ket{u}\!\bra{u})\ket{u}\right|\\
        &\geq 1-(\epsilon_N+2\delta_N)\text{.}
    \end{align*}
    By the continuity of $g_N$ the above inequality actually holds for all $\ket{u}\!\bra{u}\in \D_\text{pure}(\h)$. Choosing $N$ large enough so that $\epsilon_N+2\delta_N<1$, we define $\hat{g}_N(\ket{u}):=g_N(\ket{u})/|g_N(\ket{u})|$ and restrict it to input vectors with all elements other than the first two fixed to zero. We get a continuous $1$-homogeneous function $\hat{\mathbb{C}^2}\to \Sp^1$, which contradicts \Cref{lem_2homog}.
\end{proof}

Intuitively, the problem is the impossibility of a continuous $\vv$. To see the impossibility, consider the above proof for a $\vv$ continuous (i.e., for $\epsilon_N=\delta_N=0$). To see that this is a problem, consider $\vv_0$ of \cref{eq:v_example}. For the $N$-th state tomography algorithm of \cref{eq:tom_tr} take any $\epsilon\leq \epsilon_N^{\tr}/4$ and
\begin{eqnarray*}
    \rho = \begin{pmatrix}
        \epsilon^2 & -\epsilon\sqrt{1-\epsilon^2}\\
        -\epsilon\sqrt{1-\epsilon^2} & 1-\epsilon^2
    \end{pmatrix}\text{, }&&
    x = \begin{pmatrix}
        \epsilon^2 & \epsilon\sqrt{1-\epsilon^2}\\
        \epsilon\sqrt{1-\epsilon^2} & 1-\epsilon^2
    \end{pmatrix}\text{,}
\end{eqnarray*}
and suppose that on input $\rho$ the algorithm outputs $x$ with high probability. This supposition is consistent with state tomography, because $||x-\rho||_{\tr}\leq \epsilon_N^{\tr}$. However,  we have $||\vv_0(x)-\vv_0(\rho)||_2=2\sqrt{1-\epsilon^2}\geq 2-\epsilon_N^{\tr}/2$, because $\rho=\ket{u}\!\bra{u}$ is too close to the discontinuity of $\vv_0$ (corresponding to $|\left\langle 0 | u \right\rangle | = \epsilon$ in \Cref{fig:tomography_plot}).  

To fix the problem, we would like to avoid the discontinuity in a way that still accepts completely unrestricted and unknown $\rho\in\D_\text{pure}(\h)$. The idea is to approximately locate each input $\rho$ and \emph{then} to choose a $\vv$ function that is continuous there. Instead of a single $\vv_0$, we consider multiple $\vv_i$, so that for every $\rho\in\D_\text{pure}(\h)$ there is $i\in\{0,1,\dots, d-1\}$ such that $\rho$ is far from the discontinuity of $\vv_i$. We suggest
\begin{equation}
    \vv_i(\rho) =  \frac{\rho\ket{i}}{\sqrt{\bra{i}\rho\ket{i}}}\text{.}\label{eq:v_i}
\end{equation}
Having sampled the estimate $x$ we apply $\vv_{r(x)}$, where $r(x)=\min\{i: \bra{i}x\ket{i}\geq 1/d\}$ determines which $\vv$ function to use\footnote{Other $r(x)$ functions are possible, as long as they are deterministic functions into $\{0, 1,\dots, d-1\}$. For example, $\arg\max_i \bra{i}x\ket{i}$ requires an additional tie-breaking rule, e.g. $r(x)=\min\{\arg\max_i \bra{i}x\ket{i}\}$.}. 
This way, the chosen vector function is always continuous around the sampled $x$, circumventing the impossibility and motivating the following modification.

\begin{mydef}[Vector tomography, revised]\label{def_sttom2}
A sequence of $ T_N$-algorithms is a \emph{vector tomography} if  $T:=\bigcup_{N=0}^{\infty}  T_N$ is dense in $\D_\text{pure}(\h)$, if there exists an indexed set $\{\vv_i\}_{i\in I}$ of maps $\vv_i: T\to \hat{\h}$ such that $\vv_i(\rho)\vv_i(\rho)^\dagger=\rho$, and if the $N$-th algorithm works as follows: Its quantum part on input $\rho^{\otimes N}$ where $\rho\in \mathcal{D}_\text{pure}(\h)$ outputs a measurement outcome $x\in T_N$ with probability $p_{N,\rho}(x)$, its classical part then calculates $r(x)\in I$ and $\vv_{r(x)}(x)$, and its probability distribution $p_{N,\rho}$ for all $\rho\in T$ satisfies
\begin{equation}
    Pr_{x\sim p_{N,\rho}(\cdot)}\left[\left|\left|\vv_{r(x)}(x)-\vv_{r(x)}(\rho)\right|\right|_2\leq \epsilon_N\right]\geq 1-\delta_N\text{,}\label{eq_sttom2}
\end{equation}
with $\lim_{N\to\infty}\epsilon_N=0$ and $\lim_{N\to\infty}\delta_N=0$.
\end{mydef}

While \Cref{def_sttom} yields a trace-preserving protocol, \Cref{def_sttom2} yields a random protocol, as we show next. This time the fully quantum version of the $N$-th algorithm needs an additional register to store $r(x)$. Before the deferred measurement, we have
    \begin{align}
        \sum_i\Psi_i(\rho^{\otimes N})\otimes\ket{i}\!\bra{i}=&\sum_{x\in T_N}p_{N,\rho}(x)\ket{e_{\vv_{r(x)}(x)}}\!\bra{e_{\vv_{r(x)}(x)}}\otimes\ket{r(x)}\!\bra{r(x)}\nonumber\\
        \Psi_i(\rho^{\otimes N})=&\sum_{x\in T_N}p_{N,\rho}(x, i)\ket{e_{\vv_{i}(x)}}\!\bra{e_{\vv_{i}(x)}}\text{,}
    \end{align}
\noindent where the joint probability distribution of the random variables $x$ and $i$ is $p_{N,\rho}(x, i)=\delta_{i, r(x)}p_{N,\rho}(x)$, because the conditional probability of $i$ given $x$ is deterministic and equals the Kronecker delta function $\delta_{i, r(x)}$. The second line follows from the first by measuring the last register in the computational basis (the deferred measurement that makes the protocol random). This time, the error is quantified by 
\begin{align}
    \sum_i\left|\left|
    h\left(\Psi_i(\rho^{\otimes N})\right)-p_{N,\rho}(i)\vv_i(\rho)\right|\right|_2 & =\sum_i\left|\left|\sum_{x\in T_N}p_{N,\rho}(x, i)
    \vv_i(x)-p_{N,\rho}(i)\vv_i(\rho)\right|\right|_2\nonumber\\
    & \leq\sum_i\sum_{x\in T_N}p_{N,\rho}(x, i)\left|\left|
    \vv_i(x)-\vv_i(\rho)\right|\right|_2\nonumber\\
    & =\sum_{x\in T_N}p_{N,\rho}(x)\left|\left|
    \vv_{r(x)}(x)-\vv_{r(x)}(\rho)\right|\right|_2\nonumber\\
    & \leq \epsilon_N + 2\delta_N\text{,}\label{eq_close_i}
\end{align}
where the second line expands the marginal $p_{N,\rho}(i)=\sum_{x\in T_N}p_{N,\rho}(x, i)$ and applies the triangle inequality, the third is due to the Kronecker delta function inside $p_{N,\rho}(x, i)$, and the fourth follows from condition \eqref{eq_sttom2} by splitting $x$ to “good” and “bad”. Satisfying the inequality \eqref{eq_close_i}, the quantum instrument $\{\Psi_i\}_i$ constitutes a random protocol as claimed. Thus, for the existence of a random protocol for vector tomography (Result \hyperref[res3]{3}), it suffices to find an instance of \Cref{def_sttom2}.

\begin{mytheorem}\label{thm_tom2}
    Any state tomography (satisfying \cref{eq:tom_tr}) gives a vector tomography of \Cref{def_sttom2} for the appropriate choice of $\vv$ functions.
\end{mytheorem}

\Cref{app_possibilities_randomness} contains the proof and the application to the superposition task. 
The proof hinges on the fact that a small trace-norm distance of pure states $\left|\left|\ket{w}\!\bra{w}-\ket{u}\!\bra{u}\right|\right|_{\tr}$ implies a small Euclidean-norm distance of vectors $\left|\left|\ket{w}-\langle u |w \rangle \ket{u}\right|\right|_2$. This is enough, because, as opposed to \Cref{def_sttom}, \Cref{def_sttom2} does not require removing the $w$-dependence from $\langle u |w \rangle \ket{u}$ (i.e., the $x$-dependence from $\vv_{r(x)}(\rho)$). In \Cref{app_possibilities_randomness} we use this vector tomography for the superposition task, obtaining a random protocol:

\begin{mycor}[Random superposition protocol]\label{thm_superposition_i}
    Let $\alpha,\beta\in\mathbb{R}_+$. Let $\Srs{ij}:=\alpha\vv_i(\rho)+\beta\vv_j(\sigma)$ where $\vv_i,\vv_j$ are given by \cref{eq:v_i}. There exists a random protocol for $\Srs{ij}$, i.e., for any $\epsilon>0$ there exists $N$ and a quantum instrument $\{\Phi_{ij}\}_{ij}$ such that
    \begin{align}
        &\sum_{ij}\bigg|\bigg|\Phi_{ij}(\rho^{\otimes N}\otimes\sigma^{\otimes N}) - \tr\left[\Phi_{ij}(\rho^{\otimes N}\otimes\sigma^{\otimes N})\right]
        \frac{\Srs{ij}\!\Srsbra{ij}}{\left\langle s_{ij}{\scriptstyle(\rho,\sigma)}|s_{ij}{\scriptstyle(\rho,\sigma)}\right\rangle}\bigg|\bigg|_{\operatorname{tr}}\leq \epsilon
    \end{align}
    holds for all pairs $\rho,\sigma\in\D_\text{pure}(\h)$ such that $\Srs{ij}\neq 0$.
\end{mycor}

\section{Discussion}\label{sec:discussion}
The first consequence of the superposition impossibility follows already from the work of Oszmaniec et al. \cite{oszmaniec2016creating} and becomes more pronounced by our generalization: Quantum computation is set apart from classical computation. We can first compare it to \emph{deterministic} classical computation. Such a comparison seems natural, because the quantum state tree \cite{aaronson2004multilinear} strongly resembles the (deterministic) arithmetic formula tree \cite{Raz_2013}. The resulting contrast is in the availability of their respective building blocks: the quantum $+$ is impossible while the arithmetic $+$ is easy. This is analogous to cloning: quantum cloning is impossible, while deterministic classical cloning is easy.

A more demanding comparison is to \emph{probabilistic} classical computation. Cloning probability vectors -- equivalent to cloning diagonal density matrices -- is impossible \cite{Barnum_1996}\footnote{In addition to their result on broadcasting, ref. \protect{\cite{Barnum_1996}} shows that a set of density states is clonable only if those state's pairwise fidelities are $0$ or $1$, i.e. only if the states in the set are pairwise orthogonal or identical. This is not true for the set of diagonal density matrices.}. Thus, the task of cloning does not separate quantum and \emph{probabilistic} classical computation. The superposing task, the quantum $+$, does! Weighted addition of unknown quantum states is impossible by ref. \cite{oszmaniec2016creating} and our generalization, while weighted addition of unknown probability vectors is done by the classical probabilistic protocol in \Cref{fig:p-swap}. From this perspective, the superposition impossibility is even more distinctly quantum than the no-cloning.

Furthermore, our work has new consequences: the separations between different \emph{quantum} computational models. The random model is strictly stronger than the trace-preserving or postselection model, when asked to implement certain tasks (functions). In terms of sample complexity, these separations are infinite! The separating task has restricted inputs (a ‘partial function’) or unrestricted inputs (a ‘total function’). Consider the superposition task with inputs restricted to orthogonal pairs. Replace all occurrences of $\hat{\h}\times\hat{\h}$ in \Cref{thm_superpos} and its proof by the set of orthonormal pairs. The result is a postselection protocol no-go for orthogonal inputs. Comparing this to the one-sample random protocol of \cite{doosti2017universal} reveals the advantage of the random model on restricted inputs. 
With unrestricted inputs (primarily considered in this work), we have shown superposition and vector tomography no-go results for any sample complexity in the postselection (or trace-preserving) model. However, both tasks have at most exponential sample complexity in the random model. There, their precise sample complexities remain open.

As the second use of the topological method of ref. \cite{gavorova2024topological}, our superposition no-go elucidates the method's applicability. When no algorithm can implement a certain function (task), the topological method might be able to prove so, if both of its two steps are applicable: 1) the existence of an algorithm imposes continuity on the function, 2) topology shows that the required function cannot be continuous. We can reverse the direction, first asking which functions cannot be continuous to then find the tasks for which the topological method is relevant.
The topologically forbidden functions in this work are of the same type as those in \cite{gavorova2024topological}: sections and continuous homogeneous functions to $\Sp^1$. 
While the latter facilitate quantitative no-gos, the former, sections, provide intuition for the past tasks and a way to recognize future use cases. Specifically, the superposition and the controlled unitary \cite{gavorova2024topological} both seem to involve choosing an unphysical representative (vector or operator) for the physical input (density matrix or superoperator). No section means no representative can be chosen continuously. The method may work for other tasks that involve such choices. More generally, the tasks must be quantum. Consistent with the fact that boolean functions are always computable by a complex enough algorithm, the topological method is restricted to tasks (functions) with topologically nontrivial output spaces, e.g. density states or superoperators. In this sense, and as opposed to the polynomial method, the topological method is purely quantum.  Quantum outputs are inherent in state complexity \cite{aaronson2004multilinear}, quantum state synthesis \cite{irani2021quantum,ji2018pseudorandom,kretschmer:LIPIcs.TQC.2021.2} and unitary synthesis \cite{haferkamp2022linear,susskind2018black,bouland_et_al:LIPIcs:2020:11748}. 

\section*{Acknowledgments}
The author would like to thank Alon Dotan for helpful discussions of algebraic topology, Dorit Aharonov, Micha{\l} Oszmaniec, Itai Leigh and Michalis Skotiniotis for useful feedback. This work was supported by the Simons Foundation (Grant No. 385590), by the Israel Science Foundation (Grants No. 2137/19 and No. 1721/17), by the European Commission QuantERA grant 
  ExTRaQT (Spanish MICINN project No. PCI2022-132965), by the Ministry for Digital Transformation and of Civil Service of the Spanish Government through the QUANTUM ENIA project call, Quantum Spain project, by the European Union through the Recovery, Transformation and Resilience Plan - NextGenerationEU within the framework of the Digital Spain 2026 Agenda, by Spanish Agencia Estatal de Investigación (Project No. PID2022-141283NB-I00).

\appendix

\section{Superposition impossibility via no-section}\label{sec_alt_proof}

This section supports the intuition we gave in \Cref{sec_method}, that our impossibility relates to the nonexistence of continuous functions mapping a pure density state to some corresponding vector $\ket{u}\!\bra{u}\mapsto \ket{\widetilde{u}}$. Noting that $\Sp^{2d-1}$ is the space of unit vectors in $\mathbb{R}^{2d}$, or, alternatively, unit vectors in $\mathbb{C}^d$, we restate this formally:

\begin{mylemma}\label{thm_alg-top}
    Let $d\in\mathbb{N}$, $d\geq 2$. Let $\pi:\Sp^{2d-1}\to\mathcal{D}_\text{pure}(\mathbb{C}^d)$ be the outer product function $\pi(\left|u\right>)=\ket{u}\!\bra{u}$. There is no continuous function $s: \mathcal{D}_\text{pure}(\mathbb{C}^d)\to \Sp^{2d-1}$ such that $\pi \circ s=id$.
\end{mylemma}

\noindent The term \emph{section} refers to a continuous function that composes with a projection into the identity function $id$, as above. We use Lemma \ref{thm_alg-top} to prove this special-case version of \Cref{thm_superpos}:

\begin{mytheorem}
 Let $\alpha,\beta\in\mathbb{R}_+$ and suppose $\dim\h\geq 3$. Let any function $c:\hat{\h}\times\hat{\h}\to\mathbb{C}$, such that $S:=c^{-1}(\mathbb{C}\setminus\{0\})$ is dense in $\hat{\h}\times\hat{\h}$, determine the superposition $\Suv{c}=\alpha \left|c(u,v)\right|\ket{u}+\beta\, c(u,v)\ket{v}$. There is no continuous function $\mathcal{A}:\mathcal{D}_\text{pure}(\h)\times \mathcal{D}_\text{pure}(\h)\to \mathcal{D}_+(\h)$ such that for all $(\ket{u},\ket{v})\in S$ 
\begin{equation}
\frac{\mathcal{A}(\ket{u}\!\bra{u},\ket{v}\!\bra{v})}{\operatorname{tr}\left[\mathcal{A}(\ket{u}\!\bra{u},\ket{v}\!\bra{v})\right]} = 
    \frac{\Suv{c}\!\Suvbra{c}}{\left\langle s_{\operatorname{c}}{\scriptstyle(u,v)}|s_{\operatorname{c}}{\scriptstyle(u,v)}\right\rangle}\text{.}\label{eq:superpos_equality}
\end{equation}
\end{mytheorem}
\begin{proof}
Assume towards contradiction that such $\mathcal{A}$ exists. Let $\ket{\gamma}\in\hat{\h}$ be some fixed vector. 
Denote by $\mathcal{K}_\gamma$ the subspace of $\h$ orthogonal to $\ket{\gamma}\in\h$. The dimension of $\mathcal{K}_\gamma$ is $d=\dim\h-1\geq 2$. Next we define the function $h_\gamma: \mathcal{D}_\text{pure}(\mathcal{K}_\gamma)\to \mathcal{K}_\gamma$:
\begin{equation*}
    h_\gamma(\ket{u}\!\bra{u}):=\frac{\alpha^2+\beta^2}{\alpha\beta}\ket{u}\!\bra{u}\frac{\mathcal{A}(\ket{u}\!\bra{u},\ket{\gamma}\!\bra{\gamma})}{\operatorname{tr}\left[\mathcal{A}(\ket{u}\!\bra{u},\ket{\gamma}\!\bra{\gamma})\right]}\ket{\gamma}.
\end{equation*}
Since $\mathcal{A}$ is continuous, so is $h_\gamma$. Let $(\ket{u_n},\ket{\gamma_n})$ be a sequence in $S$ that converges to $(\ket{u},\ket{\gamma})$. Then, in particular $\ket{u_n}\xrightarrow{n\to\infty}\ket{u}$ and
\begin{equation}
    h_\gamma(\ket{u}\!\bra{u})= \lim_{n\to\infty} h_{\gamma_n}(\ket{u_n}\!\bra{u_n}) = \lim_{n\to\infty} \frac{c(u_n,\gamma_n)^*}{\left|c(u_n,\gamma_n)\right|}\left|u_n\right>\text{,}\label{eq:limit_prod}
\end{equation}
where the second equality follows by substituting \cref{eq:superpos_equality} and removing the additive terms that converge to zero because $\left<\gamma_n | u_n\right>\xrightarrow{n\to\infty} 0$. Similarly we get that $\bra{u}h_\gamma(\ket{u}\!\bra{u})=\lim_{n\to\infty}\tfrac{c(u_n,\gamma_n)^*}{\left|c(u_n,\gamma_n)\right|}$. In other words, the fractions converge to $\lambda_u:=\bra{u}h_\gamma(\ket{u}\!\bra{u})$, and since the $n$-th fraction has norm $1$, we have $|\lambda_u|=1$. \Cref{eq:limit_prod} is the limit of the product of two converging sequences, therefore it equals the product of the limits, $h_\gamma(\ket{u}\!\bra{u})=\lambda_u\ket{u}$. Since $\mathcal{K}_\gamma\approx \mathbb{C}^d$ we got a function $h_\gamma:\mathcal{D}_\text{pure}(\mathbb{C}^d)\to \Sp^{2d-1}$ that contradicts Lemma \ref{thm_alg-top}.
\end{proof}

Next we prove Lemma \ref{thm_alg-top} by first showing that $\mathcal{D}_\text{pure}(\mathbb{C}^d)$ is homeomorphic to the complex projective space $\operatorname{\mathbb{C}P}^{d-1}$ (\Cref{claim_complex-projective}), then showing that there is no section $\operatorname{\mathbb{C}P}^{d-1}\to \Sp^{2d-1}$ (\Cref{claim_alg-top}).

\begin{myclaim}\label{claim_complex-projective}
    $\mathcal{D}_\text{pure}(\mathbb{C}^d)$ is homeomorphic to the complex projective space $\operatorname{\mathbb{C}P}^{d-1}$.
\end{myclaim}

If we regard the vectors in $\Sp^{2d-1}$ as the $d$-dimensional complex vectors of norm $1$, then complex projective space is the quotient space $\operatorname{\mathbb{C}P}^{d-1}=\Sp^{2d-1}/\sim\,$ under the identification $\left|u\right>\sim e^{i\alpha}\left|u\right>$ for any real $\alpha$.

\begin{myclaim}\label{claim_alg-top}
    Let $d\in\mathbb{N}$, $d\geq 2$. Let $p:\Sp^{2d-1}\to \operatorname{\mathbb{C}P}^{d-1}$ be the quotient map. There is no continuous function $s: \operatorname{\mathbb{C}P}^{d-1}\to \Sp^{2d-1}$ such that $p \circ s=id$.
\end{myclaim}

\begin{proof}[Proof of \Cref{claim_complex-projective}]
Let $\pi:\Sp^{2d-1}\to \mathcal{D}_\text{pure}(\mathbb{C}^d)$ be the outer product function $\pi(\left|u\right>)=\ket{u}\!\bra{u}$. Since $\mathcal{D}_\text{pure}(\mathbb{C}^d)$ is in fact the image of $\pi$, $\pi$ is surjective. Let $p$ be the quotient map $p: \Sp^{2d-1}\to \operatorname{\mathbb{C}P}^{d-1}$ as above, $p(\left|u\right>)=[\left|u\right>]$, where $[\left|u\right>]$ denotes the equivalence class under the equivalence relation $\left|u\right>\sim e^{i\alpha}\left|u\right>$ for any real $\alpha$. Observe that $\pi\left(\left|u\right>\right)=\pi\left(\left|v\right>\right)$ whenever $p\left(\left|u\right>\right)=p\left(\left|v\right>\right)$. Then by the universal property of the quotient space in topology there exists $t:\operatorname{\mathbb{C}P}^{d-1}\to \mathcal{D}_\text{pure}(\mathbb{C}^d)$ continuous, such that $t\circ p=\pi$:

\begin{center}
\begin{tikzpicture}[node distance=2cm, auto]
  \node (S) {$\Sp^{2d-1}$};
  \node [right of=S] (CP) {$\operatorname{\mathbb{C}P}^{d-1}$};
  \node [node distance=1.4cm, below of=CP] (im) {$\mathcal{D}_\text{pure}(\mathbb{C}^d)$};
  \draw[->] (S) to node[above] {$\pi$} (im.north west);
  \draw[->] (S) to node[above] {$p$} (CP);
  \draw[red, thick, dashed, ->, text=red] (CP) to node {$t$} (im);
\end{tikzpicture}
\end{center}

In particular, $t$ is surjective. Next, we show that $t: [\left|u\right>]\mapsto \ket{u}\!\bra{u}$ is injective: Suppose $\ket{u}\!\bra{u}=\ket{v}\!\bra{v}=:M$. Since $\left|u\right>$, $\left|v\right>$ are unit vectors, there exists a column in matrix $M$ that is nonzero, label this column’s index by $j$. This column is equal to $M_j=u_j^*\left|u\right>=v_j^*\left|v\right>$. Since it is nonzero we can write $\left|u\right>=\frac{v_j^*}{u_j^*}\left|v\right>$, and since $\left|u\right>$, $\left|v\right>$ are unit vectors, we have $\frac{v_j^*}{u_j^*}=e^{i\alpha}$ for some real $\alpha$ and therefore $[\left|u\right>]=[\left|v\right>]$. We have shown that $t:\operatorname{\mathbb{C}P}^{d-1}\to\mathcal{D}_\text{pure}(\mathbb{C}^d)$ is continuous and bijective. The domain of $t$ is compact, because a quotient space of a compact space (here $\Sp^{2d-1}$) is compact. The range of $t$ is Hausdorff. A bijective continuous map from a compact space to a Hausdorff space is closed and hence a homeomorphism. Thus, $t$ is a homeomorphism.
\end{proof}

\begin{proof}[Proof of \Cref{claim_alg-top}]
This proof uses some standard facts in homology theory. They can be found, for example, in Chapter 2 in \cite{hatcher2000algebraic}. 
The reduced homology groups of $\Sp^m$ are
\begin{equation*}
    \tilde{H}_i\left(\Sp^m\right)=\begin{cases}
\mathbb{Z}, & i=m\\
0, & i\neq m\end{cases}.
\end{equation*}
The singular homology groups of $\operatorname{\mathbb{C}P}^{k}$ are
\begin{equation*}
    H_i\left(\operatorname{\mathbb{C}P}^{k}\right)=\begin{cases}
\mathbb{Z}, & i=0,2,4\dots, 2k\\
0, & \text{otherwise}\end{cases},
\end{equation*}
and by definition $\tilde{H}_i\left(X\right)=H_i\left(X\right)$ for all $i>0$ for any space $X$.
Assume towards contradiction that $s:\operatorname{\mathbb{C}P}^{d-1}\to \Sp^{2d-1}$ continuous exists for some $d\geq 2$ such that $p\circ s=id$. In diagram:

\begin{center}
\begin{tikzpicture}[node distance=2cm, auto]
  \node (CP1) {$\operatorname{\mathbb{C}P}^{d-1}$};
  \node [right of=CP1] (S) {$\Sp^{2d-1}$};
  \node [right of=S] (CP2) {$\operatorname{\mathbb{C}P}^{d-1}$};
  \draw[->] (CP1) to node[above] {$s$} (S);
  \draw[->] (S) to node[above] {$p$} (CP2);
  \draw[->, bend right] (CP1) to node[above] {$id$} (CP2);
\end{tikzpicture}
\end{center}

\noindent A continuous map between two spaces $f:X\to Y$ induces for any $i$ a homomorphism $f_*:H_i(X)\to H_i(Y)$ between their $i$-th singular homology groups. For $i=2d-2$ we get  

\begin{center}
\begin{tikzpicture}[node distance=3.5cm, auto]
  \node (CP1) {$H_{2d-2}\left(\operatorname{\mathbb{C}P}^{d-1}\right)$};
  \node [right of=CP1] (S) {$H_{2d-2}\left(\Sp^{2d-1}\right)$};
  \node [right of=S] (CP2) {$H_{2d-2}\left(\operatorname{\mathbb{C}P}^{d-1}\right)$};
  \draw[->] (CP1) to node[above] {$s_*$} (S);
  \draw[->] (S) to node[above] {$p_*$} (CP2);
  \draw[->, bend right] (CP1) to node[above] {$id_*$} (CP2);
\end{tikzpicture}
\end{center}

\noindent where $id_*$ is the identity map. However, by the above,  $H_{2d-2}(\operatorname{\mathbb{C}P}^{d-1})=\mathbb{Z}$ and $H_{2d-2}(\Sp^{2d-1})=0$ for $d\geq 2$, which contradicts the last diagram.
\end{proof}

\section{Possibilites via randomness (proofs)}\label{app_possibilities_randomness}

Here we prove \Cref{thm_tom2} and \Cref{thm_superposition_i}.

\begin{proof}[Proof of \Cref{thm_tom2}]
    We write $\rho=\ket{u}\!\bra{u}$ and $x=\ket{w}\!\bra{w}$ for the unknown state and the measurement outcome respectively. Suppose
    \begin{equation}
    \left|\left|x-\rho\right|\right|_{\text{tr}}=\left|\left|\ket{w}\!\bra{w}-\ket{u}\!\bra{u}\right|\right|_{\text{tr}}\leq \epsilon_N^\text{tr}< \frac{1}{d}\label{eq:state_supposition}\text{.}
    \end{equation}

\noindent We have 
\begin{align}
    \left|\left|\ket{w}-\langle u |w \rangle \ket{u}\right|\right|_2
    &= \sqrt{\left|\left|\ket{w}\right|\right|^2_2 - 2 \left|\langle u |w \rangle\langle w |u \rangle\right|
    + \left|\langle u |w \rangle \right|^2\left|\left|\ket{u}\right|\right|^2_2}\nonumber\\
    &= \sqrt{1-\left|\langle u |w \rangle\right|^2}\nonumber\\
    &\leq \frac{1}{2}\epsilon_N^\text{tr} \text{,}\label{eq:euc_from_tr}
\end{align}
where the last inequality follows from the identity
\begin{equation}
    \left|1-\left|\langle u |w \rangle\right|^2\right|
    = \frac{1}{4}\left|\left|\ket{u}\!\bra{u}-\ket{w}\!\bra{w}\right|\right|_{\text{tr}}^2
    \leq \frac{1}{4}\left(\epsilon_N^\text{tr}\right)^2\text{.}
    \label{eq:close_to_1}
\end{equation}
See, for example, \cite[eq. (9.172)]{wilde2011classical}. 

Next we use \eqref{eq:euc_from_tr} to upper bound
\begin{equation}
    \left|\left|\vv_{r(x)}(x)-\vv_{r(x)}(\rho)\right|\right|_2 = \left|\left|\ket{w}-\frac{|\langle r |u\rangle|}{|\langle r |w\rangle|}\frac{\langle r |w\rangle}{\langle r |u\rangle}\ket{u}\right|\right|_2\label{eq:v_i_rewritten}\text{.}
\end{equation}
\noindent where we define $\vv_i$ as in \eqref{eq:v_i} and we wrote $r(x)=r$ for readability. The rule for choosing $r(x)$ ensures $|\langle r | w \rangle|^2\geq1/d$ which together with \eqref{eq:state_supposition} ensures $\langle r | u \rangle\neq 0$. Comparing the left-hand side of \eqref{eq:euc_from_tr} and the right-hand side of \eqref{eq:v_i_rewritten}, we note that it is enough to compare the scalars next to $\ket{u}$. Using the closeness of the scalars that follows from \eqref{eq:euc_from_tr},
\begin{align}
    \left| \langle r |w \rangle - \langle u |w \rangle \langle r |u \rangle\right|&\leq \frac{1}{2}\epsilon_N^\text{tr}\label{eq:ineq_midstep}
\end{align}
\noindent for any $r\in\{0, 1,\dots d-1\}$, we upper bound

\begin{align}
    \left|\langle u |w \rangle- \frac{|\langle r |u\rangle|}{|\langle r |w\rangle|}\frac{\langle r |w\rangle}{\langle r |u\rangle}\right|
    &\leq 
    \left|\langle u |w \rangle-\tfrac{\langle u |w \rangle}{|\langle u |w \rangle|}\right|
    +
    \left|\tfrac{\langle u |w \rangle}{|\langle u |w \rangle|} - \langle u |w \rangle\frac{|\langle r |u\rangle|}{|\langle r |w\rangle|}\right|
    +
    \left|\langle u |w \rangle\frac{|\langle r |u\rangle|}{|\langle r |w\rangle|} - \frac{|\langle r |u\rangle|}{|\langle r |w\rangle|}\frac{\langle r |w\rangle}{\langle r |u\rangle}\right|\nonumber\\
    &= \left||\langle u |w \rangle|-\tfrac{|\langle u |w \rangle|}{|\langle u |w \rangle|}\right|
    +
    |\langle u |w \rangle| \left|\tfrac{1}{|\langle u |w \rangle|} - \frac{|\langle r |u\rangle|}{|\langle r |w\rangle|}\right|
    +
    \frac{|\langle r |u\rangle|}{|\langle r |w\rangle|}\left|\langle u |w \rangle - \frac{\langle r |w\rangle}{\langle r |u\rangle}\right|\nonumber\\
    &\leq \left||\langle u |w \rangle|-1\right| + |\langle u |w \rangle|\frac{\frac{1}{2}\epsilon_N^\text{tr}}{|\langle u |w \rangle|\,|\langle r |w\rangle|} + \frac{|\langle r |u\rangle|}{|\langle r |w\rangle|}\frac{\frac{1}{2}\epsilon_N^\text{tr}}{|\langle r |u\rangle|}\label{eq:last_line_st}\text{,}
\end{align}
\noindent where, in addition to \eqref{eq:ineq_midstep}, for the middle term in the last line we used also $|\,|a| - |b|\,|\leq |a-b|$ $a,b\in\mathbb{C}$. Applying to the first term $|1-|\langle u |w \rangle||\leq |1-|\langle u |w \rangle|^2|$ and \eqref{eq:close_to_1}, and simplifying the other terms, we get
\begin{equation}
    \left|\langle u |w \rangle- \frac{|\langle r |u\rangle|}{|\langle r |w\rangle|}\frac{\langle r |w\rangle}{\langle r |u\rangle}\right|\leq \frac{1}{4}\left(\epsilon_N^{\text{tr}}\right)^2 + \frac{\epsilon_N^\text{tr}}{|\langle r |w\rangle|}\text{.}\label{eq:scalart_close_st}
\end{equation}

\noindent We obtain
\begin{equation}
\left|\left|\vv_{r(x)}(x)-\vv_{r(x)}(\rho)\right|\right|_2 \leq \left(\frac{1}{|\langle r |w\rangle|}+\frac{1}{2}\right)\epsilon_N^{\text{tr}}
+ \frac{1}{4}\left(\epsilon_N^{\text{tr}}\right)^2\label{eq:last_euc}
\end{equation}

\noindent by combining \eqref{eq:euc_from_tr} and \eqref{eq:scalart_close_st}. 

The following parameters then satisfy \Cref{def_sttom2}. Set $\delta_N=\delta_N^\text{tr}$. Let $N_0$ be such that $\epsilon_{N}^\text{tr}\leq 1/d$ for all $N\geq N_0$, and set

\begin{equation}
    \epsilon_{N}:=  
    \begin{cases}
    2 & \text{if } N<N_0\\
    \left(\sqrt{d\,}+\frac{1}{2}\right)\epsilon_N^\text{tr} + \frac{1}{4}\left(\epsilon_N^{\text{tr}}\right)^2 & \text{if } N\geq N_0
    \end{cases}\text{.}
\end{equation}
Since the $r(x)$ function ensures $|\left\langle r|w\right\rangle|^2\geq 1/d$ for $N\geq N_0$, this $\epsilon_N$ upper bounds \eqref{eq:last_euc}. Moreover, $\lim_{N\to\infty}\epsilon_N=0$ and  $\lim_{N\to\infty}\delta_N=0$, satisfying \Cref{def_sttom2}.

\end{proof}

\begin{proof}[Proof of \Cref{thm_superposition_i}]
    We suggest the following random protocol $\{\Phi_{ij}\}_{ij}$. First, it performs a vector tomography on each of the two input states. The classical estimate of the desired superposition is then $\Sxy{ij}$ with probability $p_{N,\rho, \sigma}(x,i, y,j)$,
    where $x$ and $y$ are the estimates of $\rho$ and $\sigma$. 
    Assume for now that whenever the probability is nonzero, we have $||\Sxy{ij}||_2^2\geq K$ for some positive constant $K$. We will justify this later. Having the classical description of $\Sxy{ij}$, the protocol runs a circuit that prepares the renormalized $\Sxy{ij}$ within the error $\epsilon/4$. In other words, $\Phi_{ij}$ satisfies:
    \begin{equation}
        \left|\left|\Phi_{ij}(\rho^{\otimes N}\otimes\sigma^{\otimes N}) - \sum_{x,y\in T_N}p_{N,\rho, \sigma}(x,i, y,j)
        \frac{\Sxy{ij}\!\Sxybra{ij}}{\left\langle s_{ij}{\scriptstyle(x,y)}|s_{ij}{\scriptstyle(x,y)}\right\rangle}\right|\right|_{\operatorname{tr}}\leq \epsilon/4
    \end{equation}
    and, consequently
    \begin{equation}
        \left|\left|\tr\left[\Phi_{ij}(\rho^{\otimes N}\otimes\sigma^{\otimes N})\right] - \sum_{x,y\in T_N}p_{N,\rho, \sigma}(x,i, y,j)\right|\right|_{\operatorname{tr}}\leq \epsilon/4\text{.}
    \end{equation}
    Thus, to prove \Cref{thm_superposition_i} it remains to show that
\begin{equation}
    D:=\sum_{i, j}\left|\left|\sum_{x,y\in T_N}
    p_{N,\rho, \sigma}(x,i, y,j)\left(
    \frac{\Sxy{ij}\!\Sxybra{ij}}{\left\langle s_{ij}{\scriptstyle(x,y)}|s_{ij}{\scriptstyle(x,y)}\right\rangle} - 
    \frac{\Srs{ij}\!\Srsbra{ij}}{\left\langle s_{ij}{\scriptstyle(\rho,\sigma)}|s_{ij}{\scriptstyle(\rho,\sigma)}\right\rangle}\right)
    \right|\right|_{\operatorname{tr}}\leq \epsilon/2\text{.}
\end{equation}

\noindent Towards proving the above, we first show that
\begin{align}
    D
    &
    \leq \sum_{i, j}\sum_{x,y}p_{N,\rho, \sigma}(x,i, y,j)\left|\left|
    \frac{\Sxy{ij}\!\Sxybra{ij}}{\left\langle s_{ij}{\scriptstyle(x,y)}|s_{ij}{\scriptstyle(x,y)}\right\rangle} - 
    \frac{\Srs{ij}\!\Srsbra{ij}}{\left\langle s_{ij}{\scriptstyle(\rho,\sigma)}|s_{ij}{\scriptstyle(\rho,\sigma)}\right\rangle}
    \right|\right|_{\operatorname{tr}}\nonumber\\
    &
    \leq \sum_{i, j, x, y}p_{N,\rho, \sigma}(x,i, y,j)\frac{2}{||\Sxy{ij}||^2_2}\left|\left|
    \Sxy{ij}\!\Sxybra{ij} - 
    \Srs{ij}\!\Srsbra{ij}
    \right|\right|_{\operatorname{tr}}\nonumber\\
    &
    \leq\frac{4\left(\alpha +\beta \right)}{K} \sum_{i, j,x,y}p_{N,\rho, \sigma}(x,i, y,j)\left|\left|
    \Sxy{ij} - 
    \Srs{ij}
    \right|\right|_2\nonumber\\
    &
    \leq \frac{4\left(\alpha +\beta \right)}{K}\left(\alpha \sum_{i,x}p_{N,\rho,\sigma}(x,i)\left|\left|
\vv_i(x) - \vv_i(\rho)\right|\right|_2  + \beta \sum_{j,y}p_{N,\rho,\sigma}(y,j)\left|\left|
\vv_j(y) - \vv_j(\sigma)\right|\right|_2\right)\nonumber\\
&
\leq \frac{4\left(\alpha +\beta \right)^2}{K}\left(\epsilon_N+2\delta_N\right)\text{,}\label{eq_corr_main}
\end{align}
where the second inequality holds because of the identity 
\begin{equation}
    \left|\left| \frac{A}{\left|\left|A\right|\right|} - 
    \frac{B}{\left|\left|B\right|\right|}
    \right|\right|
    \leq \frac{2}{\left|\left|A\right|\right|}\left|\left| A - B \right|\right|\text{,}\label{eq:renormalized_vs_not}
\end{equation}
which holds for any norm and operators $A, B$, and can be proven by substracting and adding $B/||A||$, the third inequality follows from our assumption that either $p_{N,\rho}(x,i)\,
    p_{N,\sigma}(y,j)=0$ or $||\Sxy{ij}||^2_2\geq K$ and from the identity 
\begin{equation}
    ||\ket{u}\!\bra{u}-\ket{w}\!\bra{w}||_{\tr}\leq 2\max(||\ket{u}||,||\ket{w}||)||\ket{u}-\ket{w}||_2\text{,}\label{eq:tr_euc}
\end{equation}
the fourth inequality follows from the triangle inequality, and the fifth inequality holds assuming the two vector tomographies satisfy \eqref{eq_close_i}.

Next we ensure that $K$ exists. If $\alpha\neq\beta$ then we can set $K=|\alpha-\beta|^2$, because
\begin{equation}
    ||\Sxy{ij}||_2=||\alpha\vv_i(x)+\beta\vv_j(y)||_2\geq|\alpha-\beta|\text{,}
\end{equation}
by the triangle inequality. This $K$ is indeed independent of $N$, $x$, $y$. If $\alpha=\beta$, then we need another way to make sure that a small-norm $\Sxy{ij}$ has probability zero, $p_{N,\rho, \sigma}(x,i, y,j)=0$. In the vector tomography we have seen so far, for a given $x$ all $i$-values except one have probability zero, because $i$, the index of the $\vv$-function, is chosen deterministically from $x$ by the rule $i=r(x)=\min\{l: \bra{l}x\ket{l}\geq 1/d\}$. For $j$, we will use a new rule $j=r'(x,y)$:
\begin{equation}
    r'(x, y):=  
    \begin{cases}
    r(x) & \text{if } \left|\left| x - y \right|\right|_{\operatorname{tr}} < \frac{1}{2d} \\
    \min\{l: \bra{l}y\ket{l}\geq 1/d\} & \text{otherwise}\label{eq:r_xy}
    \end{cases}
\end{equation}

\noindent \Cref{def_sttom2} is satisfied also with this rule (replace $d$ with $2d$ in the proof of \Cref{thm_tom2}).

Within the case $\alpha=\beta$, suppose that the two state tomographies output $x$, $y$ such that $\left|\left| x - y \right|\right|_{\tr} < \frac{1}{2d}$. The rule $r'$ ensures that $p_{N,\rho, \sigma}(x,i, y,j)$ is nonzero only for $\Sxy{ij}=\Sxy{ii}=\alpha(\vv_i(x) + \vv_i(y))$ and
\begin{align}
    \left|\left|\Sxy{ii} \right|\right|_2=\alpha\left|\left|\vv_i(x) + \vv_i(y) \right|\right|_2
    &\geq \alpha\left( 2- \left|\left|\vv_i(x)-\vv_i(y)\right|\right|_2\right) \nonumber\\
    &\geq \alpha\left( 2-\frac{2}{\sqrt{\langle i| x | i \rangle}}\left|\left|x\ket{i}-y\ket{i}\right|\right|_2\right) \nonumber\\
    &\geq \alpha\left(2-\frac{2}{\sqrt{\langle i| x | i \rangle}}\left|\left|x-y\right|\right|_{\tr}\right) \nonumber\\
    &\geq \alpha\left(2-\sqrt{d}\frac{2}{\sqrt{2d}}\right)=\alpha(2-\sqrt{2})\text{,}\label{eq_vecdiff}
\end{align}
where we applied the triangle inequality in the first line, inequality \eqref{eq:renormalized_vs_not} in the second line, the contractivity of the trace norm in the third line, and the $r$ and $r'$ rules in the fourth line.
The resulting lower bound is our $\sqrt{K}$ for this case.

Now suppose $\alpha=\beta$ and $x,y$ are such that $\left|\left| x - y \right|\right|_{\tr} \geq \frac{1}{2d}$. We have 
\begin{equation}
\left|\left|\Sxy{ij} \right|\right|_2=\alpha \left|\left|\vv_i(x)+\vv_j(y)) \right|\right|_2 
\geq 
\frac{\alpha}{2}\left|\left| x - y \right|\right|_{\tr} \geq \frac{\alpha}{4d}\text{,}
\end{equation}
\noindent which is an application of \eqref{eq:tr_euc}. This lower bound is our $\sqrt{K}$ for this final case. Taking the minimum $K$ over all the cases, we choose $N$ large enough so that \eqref{eq_corr_main} is upper bounded by $\epsilon/2$. This completes the proof.
\end{proof}

\end{document}